\def\year{2019}\relax
\documentclass[letterpaper]{article} %
\usepackage{aaai19}  %
\usepackage{times}  %
\usepackage{helvet}  %
\usepackage{courier}  %
\usepackage{url}  %
\usepackage{graphicx}  %

\usepackage{mathtools,amsthm,amsfonts,amssymb}
\usepackage{natbib}
\usepackage{hyperref}
\usepackage[capitalize]{cleveref}
\usepackage{microtype}

\newtheorem{theorem}{Theorem}[section]
\newtheorem{observation}[theorem]{Observation}
\newtheorem{lemma}[theorem]{Lemma}

\nocopyright
\allowdisplaybreaks

\newcommand{\prof}{\pi}
\newcommand{\prob}{p}
\newcommand{\capa}{q}
\newcommand{\loca}{\ell}
\newcommand{\job}{j}
\newcommand{\jobn}{k}
\newcommand{\coin}{c}
\newcommand{\agen}{i}
\newcommand{\Agen}{A}
\newcommand{\Grou}{S}
\newcommand{\GrouAlt}{T}
\newcommand{\grou}{x}
\newcommand{\groualt}{y}
\newcommand{\func}{f}
\newcommand{\ofunc}{o}
\newcommand{\tof}{\tilde{o}}
\newcommand{\tfunc}{\tilde{f}}
\newcommand{\prog}{p}
\newcommand{\sett}{s}
\newcommand{\corr}{C}

\newcommand{\step}{\mapsto}
\newcommand{\eval}{\Downarrow}

\newcommand{\argmin}{\mathrm{argmin}}
\newcommand{\smhat}[1]{\smash{\hat{#1}}}
\newcommand{\glink}[1]{{\small\url{#1}}}

\begin{document}
\title{Migration as Submodular Optimization}
\author{Paul G\"olz \and Ariel D. Procaccia\\
Computer Science Department\\
Carnegie Mellon University}
\maketitle
\begin{abstract}
Migration presents sweeping societal challenges that have recently attracted significant attention from the scientific community. One of the prominent approaches that have been suggested employs optimization and machine learning to match migrants to localities in a way that maximizes the expected number of migrants who find employment. However, it relies on a strong additivity assumption that, we argue, does not hold in practice, due to competition effects; we propose to enhance the data-driven approach by explicitly optimizing for these effects. Specifically, we cast our problem as the maximization of an approximately submodular function subject to matroid constraints, and prove that the worst-case guarantees given by the classic greedy algorithm extend to this setting. We then present three different models for competition effects, and show that they all give rise to submodular objectives. Finally, we demonstrate via simulations that our approach leads to significant gains across the board.  \end{abstract}

\section{Introduction}\label{sec:intro}
Migration is one of the greatest societal challenges facing humanity in the 21st Century.
In 2015 there were 244 million international migrants in the world, suggesting a larger increase in the rate of international migration than was previously anticipated~\citep{MR18}.
A key reason behind this increase is the widely understood fact that migration often has significant benefits to the migrants and their families, as well as the host country and the country of origin.
For example, migrants working in the United States earn wages that are higher by a median factor of 4.11 than those they would have earned in their home countries~\citep{CMP08}; and the money they send back to their countries of origin is a reliable source of foreign currency.
But the increase in the rate of migration is also driven by globalization, social disparities, and, in no small part, by a vast number of refugees --- including, most recently, millions who have fled war in Syria, persecution in Myanmar, and economic calamity in Venezuela.

These events have fueled a surge of attempts to put migration, and, especially, refugee resettlement, on scientific footing.
Alvin Roth, a Nobel laureate in economics, nicely summarizes the problem in a 2015 opinion piece~\citep{Roth15}:
\begin{quote}
``Refugee relocation is what economists call a matching problem, in the sense that different refugees will thrive differently in different countries.
Determining who should go where, and not just how many go to each country, should be a major goal of relocation policy.''
\end{quote}
This observation underlies work in market design, which draws on the extensive literature on matching problems such as school choice. It focuses on refugee resettlement mechanisms that elicit refugees' preferences over localities, and output a matching that is desirable with respect to these preferences~\citep{MR14,JT18,DKT16}.

By contrast, a recent paper by \citet{BFHD+18}, published in \emph{Science}, takes a markedly different, data-driven approach to refugee resettlement, which employs machine learning and optimization.
Their goal is to maximize the expected number of refugees who find employment.\footnote{To be precise, they consider families and maximize the expected number of families such that at least one member is employed, but this does not make a significant technical difference, so we focus on individuals for ease of exposition.}
Using historical data from the United States and Switzerland, they train predictors that estimate the probability $\prob_{\agen \loca}$ that a given refugee $\agen$ would find employment in a given locality $\loca$.
The optimal solution is then an assignment of refugees to localities that maximizes the sum of probabilities, subject to capacity constraints.
Assuming their predictions of employment probabilities are accurate, \citeauthor{BFHD+18} demonstrate that their approach leads to a 40\%--70\% increase in the number of employed refugees, compared to the actual outcomes in the United States and Switzerland.

On a high level, we subscribe to the data-driven approach, and believe that the assumptions made by \citet{BFHD+18} are quite reasonable in the context of refugee resettlement --- most notably, the implicit assumption that the probability $\prob_{\agen \loca}$ can be estimated based only on information about the refugee $\agen$ and the locality $\loca$, and does not depend on where other refugees are assigned.
But as this approach gains traction, we envision it being deployed on a larger scale, and ultimately informing international migration policy more broadly, especially in the context of labor migration.

The key observation behind our work is that, at that scale, \emph{competition effects} would invalidate the foregoing assumption. Indeed, the larger the number of, say, engineers, who settle in a specific area, the less likely it is that any particular engineer would find a job. The immigration of approximately 331,000 Jews from the former Soviet Union to Israel in 1990 and 1991 serves as a case in point. A disproportionate number of these highly-educated immigrants were engineers and doctors, leading to a saturation of the local job markets: ``a state with an oversupply of doctors could not possibly
double the number of its doctors in several years''~\citep{Smoo08}. Our goal in this paper, therefore, is to enhance the data-driven approach to migration by explicitly modeling competition effects, and directly optimizing for them.

\subsection{Our Approach and Results}

On a technical level, our objective function $\func$ receives a set of migrant-locality pairs as input, and returns the predicted number of employed migrants under the corresponding assignment. Crucially, we assume that $\func$ is (monotone) \emph{submodular}: individual elements provide diminishing marginal returns.
Formally, if $\Grou$ and $\GrouAlt$ are two subsets of migrant-locality pairs such that $\Grou \subseteq \GrouAlt$, and $(\agen,\loca)$ is a migrant-locality pair such that $(\agen,\loca)\notin \GrouAlt$, then
$$\func(\Grou \cup \{(\agen,\loca)\}) - \func(\Grou)\geq \func(\GrouAlt \cup \{(\agen,\loca)\}) - \func(\GrouAlt).$$
This captures the idea that the larger the number of migrants that compete with $\agen$ for a job at locality $\loca$, the less likely it is that $\agen$ herself would find employment --- especially if it is a skilled job --- and the smaller the marginal contribution of $(\agen,\loca)$ to the objective function (overall number of migrants who find employment).

We can therefore cast our optimization problem as the maximization of a monotone submodular function subject to \emph{matching constraints}, which represent caps on the number of migrants each locality can absorb.\footnote{Capacity constraints can be transformed into matching constraints on the complete bipartite graph with migrants on one side, and localities on the other, where the number of copies of each locality is equal to its capacity.}
In fact, we allow the objective function to be \emph{approximately} submodular since most ways of estimating employment for a matching will introduce noise that invalidates submodularity.

The matching constraints in question can be naturally described as the intersection of two \emph{matroids}.
In \cref{sec:methods}, we show that a simple greedy algorithm --- which is known to perform well when given access to an exactly submodular function --- gives an approximation guarantee of
$( P + 1 + \frac{4 \, \epsilon}{1 - \epsilon} \, k)^{-1}$
when maximizing a submodular function $\func$ subject to $P$ many matroids if we only have access to an $\epsilon$-approximately submodular function approximating $\func$ and if $k$ is the size of the largest set in the intersection of the matroids.
We expect this result to be of independent use.
In our setting, it gives us a guarantee --- for $P = 2$ --- on the performance of our greedy matching with respect to the optimal matching.

The submodular objective function can potentially be learned, or optimized directly, from historical data~\citep{BH18,BRS16} without any further structural assumptions.
As an alternative, in \cref{sec:models}, we propose three models of how migrants find employment, all of which induce submodular objective functions.
The purpose of these models is twofold:
First, they represent different ways in which competition effects might arise, thereby helping to understand them.
Second, in practice they may prove to be more accurate as they are defined by relatively few parameters, and are easier to train.

In \cref{sec:simulations}, we compare the employment generated by the greedy heuristic on the submodular model to the baseline approach of assuming additivity.
We find that the benefits of accounting for submodularity almost always outweigh the loss associated with using an inexact optimization technique.
Across our three models and a variety of settings, the greedy approach frequently increases employment by 10\% and more, suggesting that substantial gains can be made in practice.

\subsection{Related Work}

Our work is most closely related to that of \citet{ADF17}.
Motivated by \emph{diversity} in problems such as movie recommendation and assignment of papers to reviewers, they formulate a weighted bipartite $b$-matching problem with a specific submodular objective, which is similar to our problem under a particular instantiation of the retroactive-correction model.
They show that this problem can be formulated as a quadratic program (which may be intractable at a large scale).
They also use the greedy algorithm, but only note in passing that it would give worst-case guarantees in a degenerate special case that essentially coincides with having no constraints.
By contrast, our worst-case guarantees of Section~\ref{sec:methods} hold for any \emph{approximately} submodular objective function \emph{under capacity constraints}, or even under an arbitrary number of matroid constraints.
Another key difference is that \citeauthor{ADF17} focus on one specific objective function, whereas our results of Section~\ref{sec:models} explore several different models, which are tailored to the migration domain. Perhaps the most striking difference is more fundamental: For \citeauthor{ADF17}, diversity is an objective that is orthogonal to (additive) efficiency. By contrast, we consider diversity as an inherent part of efficiency since matchings lacking in diversity will suffer from competition effects.

A bit further afield, there is a large body of work on submodular optimization, and its applications in AI. The papers of \cite{FNW78} and \cite{HS16} are especially relevant --- we discuss their results in Section~\ref{sec:methods}. 
Applications of submodular optimization include influence in social networks~\citep{KKT03}, sensor placement~\citep{KLCV+08}, and human computation~\citep{DH10}, just to name a few.
 
\section{Preliminaries}\label{sec:prelim}
Let $N$ be a set of agents or migrants, and $L$ be a set of localities.
Each locality $\loca$ has a capacity of $\capa_{\loca} \in \mathbb{N}$, limiting the number of migrants it is willing to accept.
A matching is a set $\Grou \subseteq N \times L$ of migrant-locality pairs, such that every agent is matched to at most one locality and every locality $\loca$ to at most $\capa_{\loca}$ migrants.
Our general aim is to find matchings that maximize a given function $\func : 2^{N \times L} \to \mathbb{R}_{\geq 0}$, which assigns a utility to every matching.

In this paper, we will assume that $\func$ is either submodular or approximately submodular.
To be \emph{submodular}, $\func$ must satisfy for all $\Grou \subseteq \GrouAlt \subseteq N \times L$ and $\grou \notin \GrouAlt$ that $\func(\Grou \cup \{\grou\}) - \func(\Grou) \geq \func(\GrouAlt \cup \{\grou\}) - \func(\GrouAlt)$.
To prove submodularity, it is sufficient to show the above for $\GrouAlt$ being of the shape $\Grou \cup \{\groualt\}$ with $\groualt \notin \Grou$.
A function $\func$ is \emph{supermodular} iff $(-\func)$ is submodular, i.e., iff $\func(\Grou \cup \{\grou\}) - \func(\Grou) \leq \func(\GrouAlt \cup \{\grou\}) - \func(\GrouAlt)$ for appropriate $\Grou$, $\GrouAlt$ and $\grou$.
We will assume all submodular functions to be monotone and normalized, i.e., $\func(\emptyset) = 0$.

A function $\func$ is \emph{$\epsilon$-approximately submodular} for some $\epsilon > 0$ if and only if there is a submodular function $\smhat{\func}$ such that $(1 - \epsilon) \, \smhat{\func}(\Grou) \leq \func(\Grou) \leq (1 + \epsilon) \, \smhat{\func}(\Grou)$.
Note that $\func$ itself need not be monotone.

A (finite) \emph{matroid} is a pair $(G, \mathcal{F})$ of a finite ground set $G$ and a set $\mathcal{F} \subseteq 2^G$ of \emph{independent sets} such that
(i)~$\emptyset \in \mathcal{F}$,
(ii)~if $\Grou \in \mathcal{F}$, then every subset $\GrouAlt$ of $\Grou$ is also in $\mathcal{F}$, and
(iii)~if $\Grou, \GrouAlt \in \mathcal{F}$ such that $|\Grou| = |\GrouAlt| + 1$, there is some $\grou \in \Grou \setminus \GrouAlt$ such that $\GrouAlt \cup \{\grou\} \in \mathcal{F}$.
For any set $\Grou \subseteq G$, let its \emph{rank} $r(\Grou)$ denote the size of the largest independent subset of $\Grou$, and let its \emph{span} $\mathit{sp}(\Grou)$ be the set of all elements $\grou \in G$ such that $r(\Grou) = r(\Grou \cup \{\grou\})$.
An important class of matroids are the \emph{partition matroids} that partition the ground set $G$ into disjoint subsets $G_1, \dots, G_k$ and designate $\Grou \subseteq G$ as independent iff $|\Grou \cap G_i|$ is at most some cap $g_i$ for all $i$.
 
\section{Algorithmic Results}\label{sec:methods}
Both the constraint of matching each agent to at most one locality and the constraint induced by locality quotas easily translate into partition matroids over $N \times L$.\footnote{We can even deal with migrant-locality incompatibilities, by removing incompatible pairs from the ground set.}
In the first matroid, we restrict a matching to select at most one edge out of the set of edges incident to each agent $\agen$; in the second matroid, to at most $q_{\ell}$ edges out of the set of edges incident to each locality $\ell$.
Then, the matchings are exactly the sets that are independent in both matroids.
Therefore, maximizing employment among matchings is an instance of maximizing a submodular function subject to multiple matroid
 constraints.

To optimize for employment, we must choose a way of predicting it under a given matching.
For instance, we might use techniques from machine learning to generalize past observations about employment success.
Alternatively, we might adopt a hand-crafted model --- such as the ones presented in \cref{sec:models} --- and set only its parameters based on data.
The latter approach ensures that our predictor behaves reasonably on all inputs and requires less data for fitting.

However we choose to predict employment, the corresponding employment function will have a fairly complicated shape.
As a result, we face the obstacle that optimizing the function exactly would be impractical from a computational viewpoint.
Fortunately, \citet{FNW78} found that a simple greedy algorithm gives an approximation guarantee of $\smash{\frac{1}{P+1}}$ when optimizing a monotone, submodular function $z$ subject to $P$ matroid constraints.\footnote{\citet{LMNS10} report a polynomial-time algorithm that achieves an approximation ratio of $\frac{1}{P+\epsilon}$ for $P \geq 2$ partition matroids, where $\epsilon$ can be made arbitrarily small. However, the degree of the polynomial grows very fast as $\epsilon$ becomes smaller, making their algorithm impractical for our purposes.}
Furthermore, in practice, the same greedy algorithm has been found to perform much better than this theoretical guarantee, often giving results that are close to optimal~\citep{KKT03}.

Call the matroids $\mathcal{M}_p = (N, \mathcal{F}_p)$ for $1 \leq p \leq P$, let $\mathit{sp}^p$ denote the span with respect to $\mathcal{M}_p$, and let the intersection of the matroids be $\mathcal{F} \coloneqq \smash{\bigcap_{p=1}^{P}} \mathcal{F}_p$.
To refer to the marginal contribution of an element $j$ with respect to a set $S$, we set $\rho_j(S) \coloneqq z(S \cup \{j\}) - z(S)$.

The greedy algorithm initializes $S^0 \leftarrow \emptyset$, $N^0 \leftarrow N$ and $t \leftarrow 1$.
Then, it proceeds through the following steps, wherein $t$ denotes the number of the current iteration:
\begin{description}
\item{\textbf{Step 0.}} If $N^{t-1} = \emptyset$, stop with $S^{t-1}$ as the greedy solution.
\item{\textbf{Step 1.}} Select $i(t) \in N^{t-1}$ for which $z(S^{t-1} \cup \{i(t)\})$ is maximal, with ties settled arbitrarily.
\item{\textbf{Step 2a.}} If $S^{t-1} \cup \{i(t)\} \notin \mathcal{F}$, set $N^{t-1} \leftarrow N^{t-1} \setminus \{i(t)\}$ and return to step 0.
\item{\textbf{Step 2b.}} If $S^{t-1} \cup \{i(t)\} \in \mathcal{F}$, set $\rho_{t-1} \leftarrow \rho_{i(t)}(S^{t-1})$, $S^t \leftarrow S^{t-1} \cup \{i(t)\}$, and $N^t \leftarrow N^{t-1} \setminus \{i(t)\}$.
\item{\textbf{Step 3.}} Set $t \leftarrow t + 1$ and continue from step~0.
\end{description}

Although the greedy algorithm is an effective way of maximizing submodular functions, we are unlikely to have direct access to the submodular function we are interested in, even for the value queries required for the greedy algorithm.
If we adopt a machine-learning approach, our predictor might get reasonably close to the real employment function but is unlikely to be exactly monotone and submodular.
The same problem presents itself even for the hand-crafted models:
While we prove that each model induces a perfectly submodular function, this function is defined as the expectation over a random process.
For scenarios of nontrivial size, we can only estimate this function through repeated sampling, and the introduced noise may invalidate the monotonicity and submodularity of the function.
In both cases, we can only expect our functions to be approximately submodular.

Fortunately, the greedy algorithm still performs well if $z$ is only approximately submodular.
We adapt the classic result by \citet{FNW78} to show the following approximation guarantee.
\begin{theorem}
\label{thm:approx}
    Let $z : 2^N \to \mathbb{R}$ be $\epsilon$-approximately submodular and let $\smhat{z}$ be the underlying (monotone, normalized) submodular function, i.e., let
    \[ (1 - \epsilon) \, \smhat{z}(S) \leq z(S) \leq (1 + \epsilon) \, \smhat{z}(S) \]
    for all $S \subseteq N$. 
	Let $\mathcal{F}$ be the intersection of $P$ matroids, and let $k$ denote the size of the largest $S \in \mathcal{F}$.
	Then the greedy algorithm selects a set $S \in \mathcal{F}$ such that
    \[ \left(P + 1 + \frac{4 \, \epsilon}{1 - \epsilon} \, k\right) \, \smhat{z}(S) \geq \max_{S' \in \mathcal{F}} \smhat{z}(S'). \]
\end{theorem}
To the best of our knowledge, we are the first to give approximation bounds for optimizing approximately submodular functions subject to multiple matroid constraints.
For a single matroid, \citet{HS16} give a slightly sharper bound of $(2 + \smash{\frac{2 \, \epsilon}{1 - \epsilon}} \, k) \, \smhat{z}(S) \geq \max_{S'} \smhat{z}(S')$, but their proof makes use of strong properties that only hold if $P=1$. 
By contrast, in order to capture our matching constraints we need to take the intersection of two matroids, that is, for our application $P=2$, and the corresponding worst-case approximation guarantee is roughly $1/3$ --- assuming $\epsilon$ is sufficiently small. 
This is a realistic assumption, especially in our hand-crafted models, as standard concentration inequalities imply that sampling leads to exponentially fast convergence to the true function value.

Turning to the proof, let $U^t$ be the set of elements considered in the first $t+1$ iterations except for the element added to $S$ in iteration $t+1$, i.e., $U^t \coloneqq N \setminus N^{t}$. We make use of the following known results. 

\begin{lemma}[\citealt{FNW78}]
    \label{lem:span}
For all $t$, $U^t \subseteq \bigcup_{p=1}^{P} \mathit{sp}^p(S^t)$.
\end{lemma}
\begin{lemma}[\citealt{FNW78}]
    \label{lem:dual}
    Let $\rho_i, \sigma_i \geq 0$ be given for $i = 0, \dots, K - 1$ such that $\sum_{i=0}^{t-1}\sigma_i \leq t$ for $t = 1, \dots, K$, and the sequence $(\rho_i)_i$ decreases monotonically.
    Then,
    \[ \sum_{i=0}^{K-1} \rho_i \, \sigma_i \leq \sum_{i=0}^{K-1} \rho_i .\]
\end{lemma}

We are now ready for the theorem's proof. 

\begin{proof}[Proof of \cref{thm:approx}.]
    Let $\smhat{T}$ be a maximizing subset for $\smhat{z}$ and let $S$ be the set returned by the greedy mechanism when run on $z$.
    Let $K = |S|$ and define $R_{t-1} \coloneqq \smhat{T} \cap (U^t \setminus U^{t-1})$, $r_{t-1}=|R_{t-1}|$, for $t = 1, \dots, K$.
    Without loss of generality, all elements of the ground set can actually appear in a set $S \in \mathcal{F}$, thus $U^0 =  \emptyset$. 

    Let $\smhat{\rho}_j(S)$ be the marginal contribution of $j$ with respect to $S$ and the function $\smhat{z}$.
    If $\rho_{t-1}$ was set in the algorithm as $\rho_{i(t)}(S^{t-1})$, let $\smhat{\rho}_{t-1}$ denote $\smhat{\rho}_{i(t)}(S^{t-1})$.

    By submodularity of $\smhat{z}$, it holds that
    \begin{equation}
        \label{eq:that}
        \smhat{z}(\smhat{T}) \leq \smhat{z}(S) + \sum_{j \in \smhat{T}} \smhat{\rho}_j(S).
    \end{equation}
    We can bound the second term as
    \begin{align}
        &\sum_{j \in \smhat{T}} \smhat{\rho}_j(S) \nonumber \\
        ={}& \sum_{t=1}^K \sum_{j\in R_{t-1}} \smhat{\rho}_j(S) \nonumber \\
        \leq{}& \sum_{t=1}^{K} \sum_{j\in R_{t-1}} \smhat{\rho}_j(S^{t - 1}) \tag*{(submodularity, $S^{t-1} \subseteq S$)} \\
        \leq{}& \sum_{t=1}^{K} \sum_{j\in R_{t-1}} \frac{z(S^{t - 1} \cup \{j\})}{1 - \epsilon} - \sum_{t=1}^{K} \sum_{j\in R_{t-1}} \smhat{z}(S^{t - 1}). \nonumber \\
        \intertext{The greedy algorithm chose $i(t)$ as the element $j \in N^{t - 1} = N \setminus U^{t-1}$ with maximum $z(S^{t-1} \cup \{j\})$, thus}
        \leq{}& \cfrac{1}{1-\epsilon} \sum_{t=1\vphantom{R_{t-1}}}^{K} \sum_{j\in R_{t-1}} z(S^{t - 1} \cup \{i(t)\}) - \sum_{t=1\vphantom{R_{t-1}}}^{K} \sum_{j\in R_{t-1}} \smhat{z}(S^{t-1}) \nonumber \\
        \leq{}& \frac{1 + \epsilon}{1-\epsilon} \sum_{t=1}^{K} r_{t-1} \, \smhat{z}(S^{t}) - \sum_{t=1}^{K} r_{t-1} \, \smhat{z}(S^{t-1}) \nonumber \\
        ={}& \sum_{t=1}^K r_{t-1} \, \smhat{\rho}_{t-1} + \frac{2 \, \epsilon}{1 - \epsilon} \sum_{t=1}^{K} r_{t-1} \, \smhat{z}(S^t). \label{eq:b1c}
    \end{align}

	We claim that for all $t = 1, \dots, K$, $\smash{\sum_{i=1}^t} \frac{r_{i-1}}{P} \leq t$.
    Indeed, 
    \[ \sum_{i=1}^t r_{i-1} = |\smhat{T} \cap U^t|.\]
    Since, by \cref{lem:span}, $U^t \subseteq \smash{\bigcup_{p=1}^P} \mathit{sp}^p(S^t)$, it follows that $$\sum_{i=1}^t r_{i-1} \leq \smash{\sum_{p=1}^P} |\smhat{T} \cap \mathit{sp}^p(S^t)|.$$
    Because $\smhat{T}$ is independent in the matroid $(N, \mathcal{F}_p)$, and because the rank of $\mathit{sp}^p(S^t)$ is $t$, $|\smhat{T} \cap \mathit{sp}^p(S^t)| \leq t$.
    Thus, $\sum_{i=1}^t r_{i-1} \leq P \, t$, and the claim follows.

    From the sequence $(\smhat{\rho}_{t-1})_t$, form a new sequence $(\bar{\rho}_{t-1})_t$, where 
    \[\bar{\rho}_{t-1} \coloneqq \min_{1 \leq t' \leq t} \smhat{\rho}_{t' - 1}.\]
    Clearly, this sequence decreases monotonically and is nonnegative.
    Fix some $t$ and let $t' = \argmin_{1 \leq t' \leq t} \, \smhat{\rho}_{t' - 1}$.
    Then, we can relate $\smhat{\rho}_{t-1}$ and $\bar{\rho}_{t-1}$ as follows:
    \begin{align*}
        \smhat{\rho}_{t - 1} &= \smhat{\rho}_{i(t)}(S^{t-1}) \\
        &\leq \smhat{\rho}_{i(t)}(S^{t' - 1}) \tag*{(submodularity of $\smhat{z}$, $t' \leq t$)} \\
        &= \smhat{z}(S^{t'-1} \cup \{i(t)\}) - \smhat{z}(S^{t'-1}) \\
        &\leq \frac{z(S^{t'-1} \cup \{i(t)\})}{1-\epsilon} - \smhat{z}(S^{t'-1}) \\
        &\leq \frac{z(S^{t'-1} \cup \{i(t')\})}{1-\epsilon} - \smhat{z}(S^{t'-1}) \tag*{(greedy chose $i(t')$ over $i(t)$ in iteration $t'$)} \\
        &\leq \frac{1+\epsilon}{1-\epsilon}\,\smhat{z}(S^{t'-1} \cup \{i(t')\}) - \smhat{z}(S^{t'-1}) \\
        &= \smhat{\rho}_{t'-1} + \frac{2 \, \epsilon}{1 - \epsilon} \, \smhat{z}(S^{t'}) \\
        &\leq \smhat{\rho}_{t'-1} + \frac{2 \, \epsilon}{1 - \epsilon} \, \smhat{z}(S^{t}) \tag*{(monotonicity)} \\
        &= \bar{\rho}_{t-1} + \frac{2 \, \epsilon}{1 - \epsilon} \, \smhat{z}(S^{t}) \tag*{(def.\ of $\bar{\rho}_{t-1}$, choice of $t'$).}
    \end{align*}
    By \cref{lem:dual}, we know that
    \begin{equation}
        \label{eq:srhop}
        \sum_{t=1}^K r_{t-1} \, \bar{\rho}_{t-1} = P \, \sum_{t=1}^K \frac{r_{t-1}}{P} \, \bar{\rho}_{t-1} \leq P \, \sum_{t=1}^K \bar{\rho}_{t-1}.
    \end{equation}
    This allows us to continue \cref{eq:b1c}:
    \begin{align*}
        &\sum_{t=1}^K r_{t-1} \, \smhat{\rho}_{t-1} + \frac{2 \, \epsilon}{1 - \epsilon} \sum_{t=1}^{K} r_{t-1} \, \smhat{z}(S^t) \\
        \leq{}& \sum_{t=1}^K r_{t-1} \left(\bar{\rho}_{t-1} + \frac{2 \, \epsilon}{1 - \epsilon} \, \smhat{z}(S^t)\right) + \frac{2 \, \epsilon}{1 - \epsilon} \sum_{t=1}^{K} r_{t-1} \, \smhat{z}(S^t) \\
        \leq{}& P \, \sum_{t=1}^K \bar{\rho}_{t-1} + \frac{4 \, \epsilon}{1 - \epsilon} \sum_{t=1}^{K} r_{t-1} \, \smhat{z}(S^t) \tag*{(Eq.~\ref{eq:srhop})} \\
        \leq{}& P \, \sum_{t=1}^K \smhat{\rho}_{t-1} + \frac{4 \, \epsilon}{1 - \epsilon} \sum_{t=1}^{K} r_{t-1} \, \smhat{z}(S^t) \tag*{(def.\ of $\bar{\rho}_{t-1}$)} \\
        \leq{}& P \, \smhat{z}(S) + \frac{4 \, \epsilon}{1 - \epsilon} |\smhat{T}| \, \smhat{z}(S)\\
        \leq{}& P \, \smhat{z}(S) + \frac{4 \, \epsilon}{1 - \epsilon} k \, \smhat{z}(S).
    \end{align*}
    Combining this with \cref{eq:that} yields
    \begin{equation*}
        \smhat{z}(\smhat{T}) \leq \left(P + 1 + \frac{4 \, \epsilon}{1 - \epsilon} k\right) \, \smhat{z}(S). \qedhere
    \end{equation*}
\end{proof}
 
\section{Submodular Objectives}\label{sec:models}
We propose three models for the submodular effects of competition between migrants.
Each of them makes different assumptions about how migrants find their jobs and, thus, about how they compete with each other. In each model we are interested in the \emph{expected employment function} that it induces, i.e., the function that takes as input agent-locality pairs, and outputs the expected number of agents that find employment (in each case this function depends on parameters of the model). 

\subsection{The Retroactive-Correction Model}
The easiest of our models generalizes the additive model used by \citet{BFHD+18}, by retroactively correcting employment success for competition.
We assume that the agents $N$ can be partitioned into disjoint professions $\prof$, and that only agents of the same profession compete for jobs.
Each agent $\agen$ has a probability $\prob_{\agen \loca}$ of finding employment in locality $\loca$.

In the additive model, one can imagine each migrant's job search as a coin with a bias of this probability; the employment is the number of successful coin flips, and, therefore, the expected employment is just the sum of probabilities $\prob_{\agen \loca}$ over all matched pairs of agents $\agen$ and localities $\loca$.
In this model, however, the coin flip only simulates an agent's attempt to qualify for being hired, not her chance of actually landing a job, since the latter is influenced by other agents.
An agent might qualify by means of language acquisition, by transferring degrees, through further professional training, or by choosing a good way to present herself to prospective employers.
The actual employment at a locality $\loca$ and in a profession $\prof$ is obtained by applying a concave and monotone \emph{correction function} $\corr_{\loca \prof} : \mathbb{N} \to \mathbb{R}_{\geq 0}$ to the number of qualifying agents.
Total employment is computed by summing up employment over all localities and professions; the submodular function to optimize for is the expected employment generated by this process.
While there is no direct competition between agents of different professions, a matching algorithm cannot simply treat them in isolation since different professions have to share the locality caps.

An easy example for a correction function might be $x \mapsto \min(x, c)$ for some constant $c$.
This models the scenario where there are $c$ jobs in that locality and profession, and where the job market manages to bring all qualifying agents into work, up to the hard cap of $c$.
Other functions might slow down their growth before hitting the cap, simulating inefficiencies in the job market.

\begin{observation}
The expected employment function induced by the retroactive-correction model is submodular. 
\end{observation}

To see this, it is enough to show submodularity for a single locality $\loca$ and profession $\prof$ since a sum of submodular functions is submodular.
Because all agent-locality pairs have $\loca$ as the second component, we can identify them with their agents.
For a set $\Grou$ of agents and two other agents $\grou$ and $\groualt$, we need to show that the marginal contribution of $\grou$ is at least as high with respect to $\Grou$ as it is with respect to $\Grou \cup \{\groualt\}$.
Fix the coin flips for all agents. If $\grou$ failed, her marginal contribution is zero in both cases. If $\groualt$ failed, the marginal contribution of $\grou$ is the same whether $\groualt$ is present or not.
If both succeed, the marginal contribution must be at least as high with respect to the smaller set than with respect to the larger one by the concavity of $\corr_{\loca \prof}$.
Since employment in the case without fixed coin flips can be seen as a convex combination of employment in the cases with fixed coin tosses, submodularity is preserved.

\subsection{The Interview Model}
In the second model, we again partition agents by profession.
In contrast to the previous model, a migrant's employment success is not determined by a single hurdle of qualification but through a sequence of applications to individual jobs.
Each locality $\loca$ has a certain number $\jobn_{\loca \prof}$ of jobs for profession $\prof$, and each agent $\agen \in N$ has a probability $\prob_{\agen\loca}$ of being accepted at a particular job of their profession at locality $\loca$.

Employment is calculated individually for each locality $\loca$ and profession $\prof$.
Initially, there are $\jobn_{\loca \prof}$ many jobs available.
Then, we iterate over agents of profession $\prof$ matched with locality $\loca$ in an order selected uniformly at random.
For agent $\agen$, we throw $\prob_{\agen\loca}$-biased coins until we either hit success or until we have failed as many times as the number of available jobs.
Each of these coin tosses represents the agent applying for one of the remaining jobs, where each application succeeds with probability $\prob_{\agen\loca}$.
If one of the tosses succeeds, this agent is counted as employed, and the process continues with one less job available.
Else, the agent is not employed, and the number of available jobs remains the same for the next agent.
The utility of a matching is the sum over localities and professions of the expected number of employed agents in that locality and profession.

\begin{theorem}
\label{thm:submod}
The expected employment function induced by the interview model is submodular. 
\end{theorem}

In contrast to our other models, it is quite nontrivial to establish submodularity in the interview model. The proof of Theorem~\ref{thm:submod} is relegated to Appendix~\ref{app:submod}.

\subsection{The Coordination Model}
Our final model goes beyond the strict separation between professions, and assumes that employment for all agents in the same locality is coordinated.
Similarly to the previous model, each locality $\loca$ has a number $\jobn_{\loca}$ of jobs, and each agent $\agen$ has a certain probability $\prob_{\agen \job}$ of being compatible with a specific job $\job$.
These probabilities might be induced by a strict partition into professions, but can be more fluid with people being more or less qualified for jobs closer to or further from their areas of expertise.
Whereas, in the interview model, a migrant directly takes a job as soon as she is found eligible, we now determine compatibility between all migrants and jobs in $\loca$, and we coordinate the assignment such that employment is maximized.

To calculate employment at a locality $\loca$, we flip a coin with bias $\prob_{\agen \job}$ for each agent $\agen$ matched to $\loca$ and each job $\job$ at $\loca$.
By drawing an edge between each pair $(\agen, \job)$ whose coin flip succeeded, we obtain a bipartite graph.
Employment at this locality is the size of a maximum matching, which corresponds to the highest number of agents that can be given work.
Again, the total utility of a matching is the sum of expected employment over localities.

\begin{theorem}
\label{thm:coordination}
The expected employment function induced by the coordination model is submodular. 
\end{theorem}

The relatively simple proof of the theorem builds up on a result by \citet{rabanca16} and appears in \cref{app:coordination}.
 
\section{Simulations}\label{sec:simulations}
The approximation results of \cref{sec:methods} provide a way of respecting competition effects when matching migrants to localities.
But when is it worth adopting such an approach?
After all, besides the potential benefits, embracing submodularity also entails an increase in modeling complexity, and it forces us to abandon the efficient optimization tools that are applicable to additive optimization.

If one chooses not to deal with these drawbacks, one can always just ignore submodularity:
For each agent and locality, one may estimate the probability of her finding employment at this place in the absence of all other agents. One can then optimize the additive function induced by the aformentioned probabilities, say, by formulating the problem as an integer linear program, and hope that the result would do well under the true objective function.
The approach of \citet{BFHD+18} can be understood as doing exactly this.
Our goal is to show that --- when competition effects are indeed present, and the objective function is submodular --- accounting for submodularity is almost always the preferable choice and can lead to significantly better outcomes than the additive approach described above.
To this end, we empirically evaluate the two approaches on all three models from \cref{sec:models}.

Our simulation code is written in Python; we use Gurobi for additive optimization and IGraph for computing maximum bipartite matchings.
All code is open source and available for reproduction at \glink{https://github.com/pgoelz/migration}.
We provide additional simulation results in \cref{app:simulations}, where we also link to IPython notebooks with the exact setup of all of our experiments.

For increased performance, we reuse estimations of expected employment.
When a model is queried with a matching that puts the same agents of profession $\prof$ ---  or just the same agents in case of the coordination model --- in the same locality $\loca$ as a previous matching, the previous estimation of expected employment is used.
Since all these estimations are obtained by averaging many random simulations, this should not significantly influence the experiments.
Furthermore, since the additive algorithm is not affected, and since all final utility values are computed without memoization, this only disadvantages the greedy algorithm, strengthening our findings.

Due to the high number of parameters in all of our models, we adopt a standard setting that is applicable across all of them.
Let $N$ be a set of 100 agents, split equally between two professions, 1 and 2.
While we vary the number of localities, we distribute 100 jobs --- 50 per profession --- randomly over the localities, ensuring that each locality has at least one job.
Furthermore, we set the capacity of each locality to exactly its number of jobs.
Finally, we average 1\,000 simulations to estimate expected employment whenever our models are queried with a matching.

\begin{figure*}[htb]
\centering
\includegraphics[width=\textwidth]{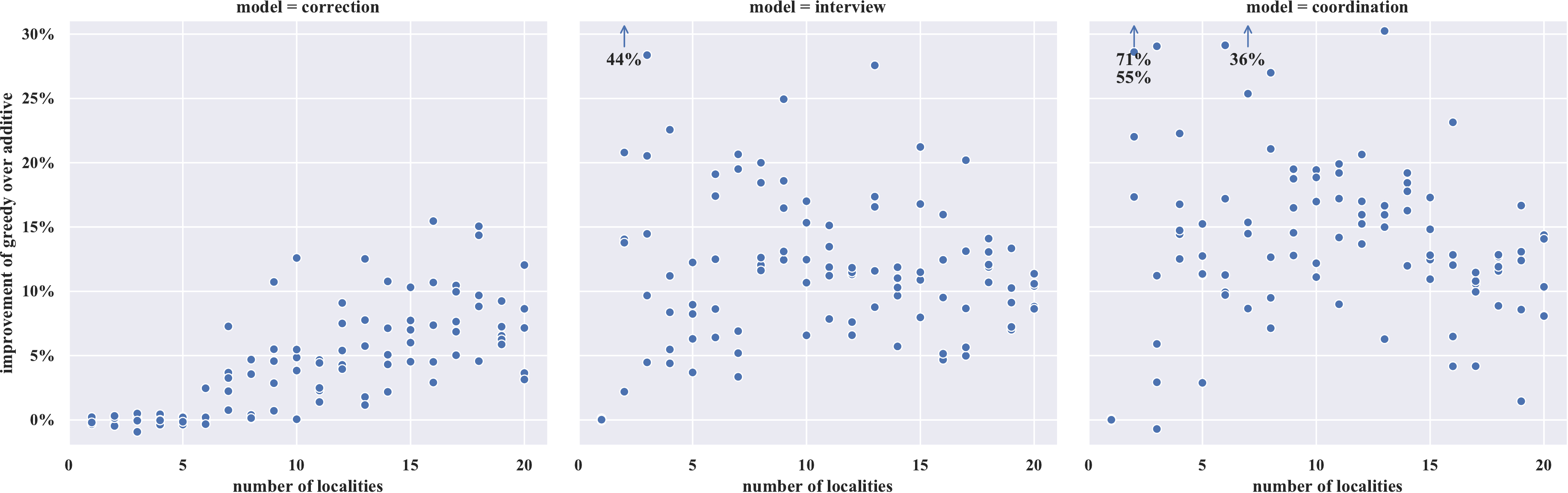}
    \caption{Increase in utility by using greedy instead of additive by number of localities. Each dot represents a new instantiation of the standard setting for the corresponding model, to which we applied both algorithms. Outliers are indicated by arrows.}
\label{fig:numlocalities}
\end{figure*}

We quickly sketch how this setting translates into the idiosyncrasies of our individual models:
In the retroactive-correction model, for each agent $\agen$, we uniformly select a probability between 0 and 1, which is used as $\prob_{\agen \loca}$ for all localities $\loca$.
We choose to keep these probabilities equal between localities to make sure that our samples will contain agents that have significantly different overall chances of being employed, an aspect that we would expect to see in any real dataset.
As the correction function for a locality $\loca$ and profession $\prof$, we choose $x \mapsto \min (x, c)$, where $c$ is the number of available jobs in $\loca$ and $\prof$.
Note that, while the cap ensures that the number of agents in a locality is at most the total number of jobs, the number of agents of a certain profession can exceed the number of jobs in that profession, which lets the cap $c$ kick in.
Likewise, in the interview model, we choose $\prob_{\agen \loca}$ uniformly at random and equal over all localities.
Finally, we induce the compatibility probabilities in the coordination model by the professions, i.e., each agent has a single compatibility probability chosen uniformly at random for all jobs of her profession and is incompatible with all other jobs.

As shown in \cref{fig:numlocalities}, the greedy algorithm outperforms additive optimization in nearly all cases, frequently increasing employment by 10\% and more.
Given the wide range of scenarios that are randomly generated, it is striking that the greedy heuristic virtually never leads to worse employment than additive maximization.
This trend persists over all settings that we simulated.
While the advantage of the greedy algorithm does not seem as pronounced for small numbers of localities in the retroactive-correction model, it leads to strong improvements over all locality numbers in the other two models.

\looseness=1
In the following, we set the number of localities to 10 and vary the number of agents $n$ --- still equally split between professions --- instead.
Each locality has a capacity of $n/10$ and has $n/10$ jobs, where the $n/2$ jobs of each profession are randomly distributed within these constraints.
\Cref{fig:numagents} shows that the greedy algorithm leads to impressive improvements across all simulated numbers of agents, again especially pronounced in the interview and coordination models.
While the gains become smaller for high numbers of agents in the correction model, the reason is innocuous: Looking at the absolute utilities instead of the ratio, we see that in these scenarios both additive and greedy get very close to an employment of half of the agents, which is the expected number of qualifying agents and thus optimal (see \cref{app:num_agents}).

\begin{figure}[htb]
\centering
    \includegraphics[width=.473\textwidth]{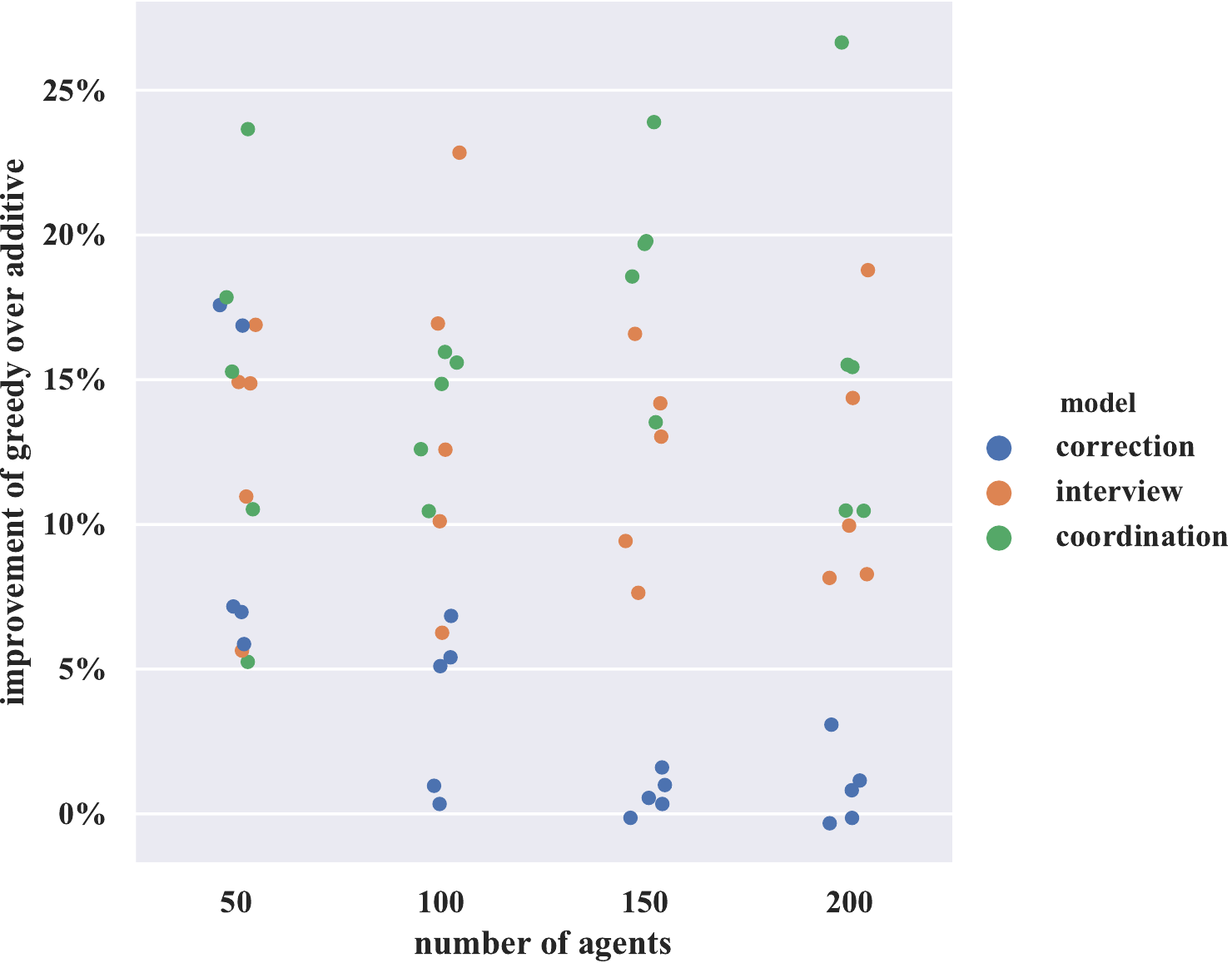}
    \caption{Improvements for different numbers of agents split across 10 localities.}
    \label{fig:numagents}
\end{figure}

In \cref{app:num_professions,app:job_availability,app:specialization}, we vary further parameters, namely the number of professions, the level of specialization of localities, and the availability of jobs with respect to fixed capacities.
In general, we find that the improvements from the greedy optimization persist across these parameter settings.

Note that the results of the greedy algorithm provide a \emph{lower bound} on what can be achieved when submodularity is taken into account.
Other heuristics with access to the submodular function might be able to further improve employment while retaining the low query complexity of the greedy algorithm.
 
\section{Discussion}\label{sec:discussion}
The simulations in the previous section revealed that employment in all of our models can be improved by accounting for submodularity.
Since this effect appeared consistently between three significantly different models of how migrants find employment, we conjecture that similar gains can be obtained in practice.
The next logical step is to measure competition effects on real migration data.

Alas, it seems to be very difficult to obtain datasets with the necessary amount of detail.
For instance, the \emph{Database on Immigrants in OECD Countries} (DIOC)\footnote{\url{https://www.oecd.org/els/mig/dioc.htm}} contains data on a vast number of migrants in different countries.
However, while File~D contains information on the level of educational attainment and on the migrants' current occupations, there is no information about the field of education or pre-migration occupation.
As a result, it is unclear whether an unemployed migrant was looking for employment at all and, if so, which job market she participated in.
An additional problem is that DIOC does not track migrants who return to their home countries as a result of unemployment.

By contrast, longitudinal studies such as the German \emph{Socio-Economic Panel}\footnote{\url{https://www.diw.de/en/diw\_02.c.222517.en/data.html}} contain extensive information about the educational background of individuals.
What is more, they follow individual migrants over an extended period of time and therefore track their success on the labor market in perfect detail.
Unfortunately, the longitudinal studies that we saw observe too few individuals spread over too many localities.
To directly observe competition, we need data on a large fraction of the migrants competing for jobs in a locality, from multiple localities.

Even the paper by \citet{BFHD+18}, for which the authors had access to sensitive migration data, conspicuously does not use any features relating to the refugees' field of occupation.
Indeed, despite the predictive power of such features for employment success in a locality, at least one major refugee-resettlement agency in the US does not track this information at all.
We have reason to believe that the same is true for other resettlement agencies as well.
A clear policy recommendation, therefore, is to compile the records in question and to make them available for research.
Hopefully, data on present migration will help to better direct flows of migration in the future --- benefitting both the migrants and the societies they join.
 
\section*{Acknowledgements}
This work was partially supported by the National Science Foundation under grants IIS-1350598, IIS-1714140, CCF-1525932, and CCF-1733556; by the Office of Naval Research under grants N00014-16-1-3075 and N00014-17-1-2428; and by a Sloan Research Fellowship and a Guggenheim Fellowship.
We would like to thank Lucas Le\-o\-pold for his helpful pointers on navigating migration datasets.

\bibliographystyle{aaai}
\bibliography{abb,ultimate,misc}

\newpage
\appendix

\onecolumn

\begin{center}
\LARGE \bf Appendix: Migration as Submodular Optimization
\end{center}

\vspace{0.3in}

\section{Proof of Theorem~\ref{thm:submod}}
\label{app:submod}

Instead of directly proving the expected employment to be submodular, it is more natural to show that the expected number of positions that remain open is supermodular.
Since a sum of supermodular functions is supermodular, it suffices to show this for a single locality $\loca$ and profession $\prof$.
Since all agent-locality pairs will have the same locality, we can treat the set of pairs as a set of agents.
Furthermore, we will show supermodularity for an arbitrary fixed order over all agents, which implies supermodularity for a random order.

For a set $\Grou$ of agents of profession $\prof$, let $\ofunc(\Grou)$ denote how many of the $\jobn_{\loca \prof}$ jobs remain open in expectation if $\Grou$ is the set of agents of profession $\prof$ matched with $\loca$, considered in the order fixed above.
We need to show that for all sets $\Grou$ of agents and distinct agents $\grou, \groualt \notin \Grou$,
\begin{equation}
\label{eq:interviewsuper}
\ofunc(\Grou \cup \{\grou, \groualt\}) - \ofunc(\Grou \cup \{\groualt\}) \geq \ofunc(\Grou \cup \{\grou\}) - \ofunc(\Grou).
\end{equation}
How the process behaves at an agent $\agen$ is entirely determined by the number $\jobn_{\agen}$ of jobs still available when it is her turn and the index $\coin_{\agen}$ of her first success in a sequence of coin flips.
More specifically, it only matters whether $\coin_{\agen} \leq \jobn_{\agen}$ (she gets a job) or $\coin_{\agen} > \jobn_{\agen}$ (she fails to get any job).
Without loss of generality, all $\coin_{\agen}$ are at most $\jobn_{\loca \prof}$ since larger values of $\coin_{\agen}$ will lead to the same behavior of never getting a job.
We show that \cref{eq:interviewsuper} holds if the expectations are conditioned on fixed values of $\coin_{\agen}$ for all agents $\agen$ up to $\grou$ or $\groualt$, whoever comes later in the order.
We will refer to all agents up to the later of $\grou$ and $\groualt$ as the \emph{prefix}, and to all later agents as the \emph{suffix}.
\Cref{eq:interviewsuper} follows from the previous proposition because the unconditional expectation is a convex combination of the expectations conditioned on these mutually exclusive and jointly exhaustive events.
Fixing the coin flips of the prefix allows us to reason about how many jobs will still be available at the first agent of the suffix, and how this differs depending on whether $\grou$ and $\groualt$ are present.
\Cref{eq:interviewsuper} will then follow from the fact that the expected number of remaining jobs is monotone and convex in the number of jobs available at the first agent of the suffix.

\subsection{Monotonicity and Convexity for a Fixed Set of Agents}
Let $\Agen$ be a set of agents $\agen$, each with their chance $\prob_{\agen\loca}$ of success when applying for a job.
Set $\tof(n)$ for the number of positions remaining open after starting from $n$ available jobs and having all agents in $\Agen$ apply in a fixed order.

\begin{lemma}
\label{lem:monoconv}
    $\tof$ is monotone and is convex (i.e., $\tof(n+2) - \tof(n+1) \geq \tof(n+1) - \tof(n)$ for all $n \in \mathbb{N}$).
\end{lemma}
\begin{proof}
    One way of visualizing the random process with expectation $\tof(n)$ is the following:
    Have a grid with $n$ columns and one row for each agent in $\Agen$, in their fixed order.
    The row corresponding to $\agen$ is filled with $\prob_{\agen\loca}$-biased coins.
    Proceed through the cells row by row, left to right.
    If the current row contains a flipped coin that was a success anywhere to the left of the current cell, remove the current coin from the grid and proceed.
    This corresponds to the case where agent $\agen$ already landed a job and does not apply for further jobs.
    Likewise, if the current column contains a successfully-flipped coin somewhere above the current cell, remove the current coin and proceed.
    This represents a job being already taken by a previous agent, which is why the current agent cannot apply there.
    Otherwise, flip the coin and proceed, simulating the current agent applying for the current job.

    Clearly, the total number of successful coin flips after proceeding through all cells can be interpreted as the resulting employment number.
    For this reason, $\tof(n)$ equals the expected number of columns without a successful coin flip.
    Note that the probabilities of different end states of the grid do not change if we proceed through the grid in a different order, as long as all cells to the left of and above a cell have been evaluated before dealing with the current cell.
    In particular, we may proceed column by column, top to bottom.

    In this interpretation, it is obvious that $\tof$ is monotone.
    Indeed, increasing its argument from $n$ to $n+1$ corresponds to adding a new column that will be evaluated after all others.
    Adding this column can increase the number of columns without successes but cannot decrease it.

    Similarly, we can now show that $\tof$ is convex.
    By linearity of expectation, the two terms $\tof(n+2) - \tof(n+1)$ and $\tof(n+1) - \tof(n)$ describe the probability that column $n+2$ and $n+1$ stays without a success, respectively.
    Fix some coin flips for the first $n$ columns.
    To stay without success, column $n+1$ must fail the coin flips in all rows that do not have a success in the first $n$ columns yet.
    In the worst case, if column $n+1$ stays without success, column $n+2$ has to satisfy the same condition.
    But if column $n+1$ has a success, $n+2$ needs to fail on one less coin to stay without success.
    This shows that, conditional on the coin flips of the first $n$ columns, column $n+2$ is at least as probable to stay without success as column $n+1$.
    The unconditional probabilities are obtained via a convex combination of the conditional probabilities, so the inequality persists.
\end{proof}

\subsection{Description of the Process with Fixed Coin Flips}

We now formalize the interviewing process for a fixed order and fixed coin flips in some prefix of agents.
A \emph{setting} consists of a sequence of integers $\coin_{\agen}$ for agents $\agen$ and a function $\tfunc : \mathbb{N} \to \mathbb{R}$.
As described earlier, $\coin_{\agen}$ is the index of the first successful coin flip for an agent $\agen$ in the prefix, capped at a large enough constant.
The function $\tfunc$ encapsulates all agents in the suffix, whose coin flips are not fixed, by giving the expected number of left-over jobs as a function of the number of jobs that are available when entering the suffix.
Such a setting can be complemented with a number $n$ of initially available jobs to form a \emph{program}.
Formally, settings $\sett$ and programs $\prog$ are defined by the following Backus-Naur forms:
\begin{align*}
    \sett &\Coloneqq [\tfunc] \mid m;\sett \\
    \prog &\Coloneqq \sett~\{n\} \mid r
\end{align*}
In the above, $n \in \mathbb{N}$ is used for the number of available jobs, $m \in \mathbb{N}_{>0}$ is used for the $\coin_{\agen}$, $r \in \mathbb{R}$ is the result of a function call, and $\tfunc$ is an externally defined function symbol $\tfunc : \mathbb{N} \to \mathbb{R}$ representing the suffix.

Programs can step according to the rules of our random process:
If the prefix is empty, we just evaluate the function $\tfunc$.
Otherwise, agent $\agen$ consumes a job iff $\coin_{\agen}$ is at most the number of available jobs.
We capture this by a stepping relation $\step$ between programs $\prog$, which is the smallest relation such that
\begin{itemize}
    \item if $r = \tfunc(n)$, then $[\tfunc]~\{n\} \step r$,
    \item if $n \geq m$, then $m;\sett~\{n\} \step \sett~\{n - 1\}$, and
    \item if $n < m$, then $m;\sett~\{n\} \step \sett~\{n\}$.
\end{itemize}

If we evaluate a program long enough, we reach a real number $r$, which is the expected number of remaining jobs after running the scenario on the number of available jobs.
To be able to directly speak about the final evaluation result of a program, we define the evaluation relation $\eval$ between programs and real numbers as the smallest one satisfying
\begin{itemize}
    \item $r \eval r$, and
    \item if $\prog \step \prog'$ and $\prog' \eval r$, then $\prog \eval r$.
\end{itemize}

\subsection{Lemmas on Evaluation with Fixed Coin Flips}
\begin{lemma}
    \label{lem:det}
    If we have $m_1; \dots; m_k; \sett~\{n\} \step^k \prog$ and $m_1; \dots; m_k; \sett'~\{n\} \step^k \prog'$, then $\prog = \sett~\{n'\}$ and $\prog' = \sett'~\{n'\}$ for a single value $n' \in \mathbb{N}$.
\end{lemma}
\begin{proof}
    By induction on $k \in \mathbb{N}$.
    If $k = 0$, the claim follows immediately for $n' = n$.
    Else, let $k > 0$.
    Then, if $n \geq m_1$, we can determine the first step of the two terms as follows:
    \begin{align*}
        &m_1; \dots; m_k; \sett~\{n\} \step m_2; \dots; m_k; \sett~\{n - 1\} \step^{k-1} \prog \\
        &m_1; \dots; m_k; \sett'~\{n\} \step m_2; \dots; m_k; \sett'~\{n - 1\} \step^{k-1} \prog'.
    \end{align*}
    The claim about $\prog$ follows from the induction hypothesis.

    Else, if $n < m_1$, we enter the analogous case of
    \begin{align*}
        &m_1; \dots; m_k; \sett~\{n\} \step m_2; \dots; m_k; \sett~\{n\} \step^{k-1} \prog \\
        &m_1; \dots; m_k; \sett'~\{n\} \step m_2; \dots; m_k; \sett'~\{n\} \step^{k-1} \prog',
    \end{align*}
    and the claim again follows from the induction hypothesis.
\end{proof}

\begin{lemma}
    \label{lem:par}
    If we have $m_1; \dots; m_k; \sett~\{n\} \step^k \prog$ and $m_1; \dots; m_k; \sett'~\{n + 1\} \step^k \prog'$, then $\prog=\sett~\{n'\}$ and $\prog'=\sett'~\{n''\}$ for $n'$ and $n''$ such that $n' \leq n'' \leq n' + 1$.
\end{lemma}
\begin{proof}
    By induction on $k \in \mathbb{N}$.
    If $k = 0$, the claim is trivial with $n' = n$ and $n'' = n+1$.
    Else, we distinguish three cases, depending on whether $m_1 \leq n$, $m_1 = n + 1$, or $m_1 > n + 1$.

    If $m_1 \leq n$, we can uniquely determine the first step of both terms as
    \begin{align*}
        &m_1; \dots; m_k; \sett~\{n\} \step m_2; \dots; m_k; \sett~\{n - 1\} \step^{k-1} \prog \\
        &m_1; \dots; m_k; \sett'~\{n + 1\} \step m_2; \dots; m_k; \sett'~\{n\} \step^{k-1} \prog',
    \end{align*}
    and the claim follows from the induction hypothesis.

    If $m_1 = n + 1$, we know that
    \begin{align*}
        &m_1; \dots; m_k; \sett~\{n\} \step m_2; \dots; m_k; \sett~\{n\} \step^{k-1} \prog \\
        &m_1; \dots; m_k; \sett'~\{n + 1\} \step m_2; \dots; m_k; \sett'~\{n\} \step^{k-1} \prog',
    \end{align*}
    and the claim follows from \cref{lem:det}.

    The case of $m_1 > n + 1$ is analogous to the first one.
\end{proof}

\subsection{Supermodularity}
We are now ready to speak about the full scenario and about how it changes depending on the presence of two prefix elements $x$ and $y$.
Throughout this section, fix a monotone, convex function $\tfunc : \mathbb{N} \to \mathbb{R}$.
Additionally, let $1 \leq x < y$, let $m_1, \dots, m_y \in \mathbb{N}_{>0}$, and let $n \in \mathbb{N}$.
Fix four programs:
\begin{align*}
    \prog_{x,y} &\coloneqq m_1; \dots; m_y; [\tfunc]~\{n\} \\
    \prog_{\bar{x},y} &\coloneqq m_1; \dots; m_{x-1}; m_{x+1}; \dots; m_y; [\tfunc]~\{n\} \\
    \prog_{x,\bar{y}} &\coloneqq m_1; \dots; m_{y-1}; [\tfunc]~\{n\} \\
    \prog_{\bar{x},\bar{y}} &\coloneqq m_1; \dots; m_{x-1}; m_{x+1}; \dots; m_{y-1}; [\tfunc]~\{n\}.
\end{align*}
Additionally, let their evaluation results be
    \[\prog_{x,y} \eval r_{x,y} \qquad
    \prog_{\bar{x},y} \eval r_{\bar{x},y} \qquad
    \prog_{x,\bar{y}} \eval r_{x,\bar{y}} \qquad
    \prog_{\bar{x},\bar{y}} \eval r_{\bar{x},\bar{y}}. \]

\begin{lemma}
    \label{lem:xy}
    It holds that $r_{x,y} - r_{\bar{x},y} \geq r_{x,\bar{y}} - r_{\bar{x},\bar{y}}$ and $r_{x,y} - r_{x,\bar{y}} \geq r_{\bar{x},y} - r_{\bar{x},\bar{y}}$.
\end{lemma}

\begin{proof}[Proof of Lemma~\ref{lem:xy}]
	It is sufficient to prove that $r_{x,y} - r_{\bar{x},y} \geq r_{x,\bar{y}} - r_{\bar{x},\bar{y}}$ because we can obtain the inequality $r_{x,y} - r_{x,\bar{y}} \geq r_{\bar{x},y} - r_{\bar{x},\bar{y}}$ by subtracting $r_{x,\bar{y}}$ and adding $r_{\bar{x},y}$.

    By \cref{lem:det}, the first $x - 1$ steps lead to settings of the same shape and a common integer $n'$.
    Thus, we may assume without loss of generality that $x = 1$.

    If $n < m_1$, $\prog_{x,y} \step \prog_{\bar{x},y}$ and $\prog_{x,\bar{y}} \step \prog_{\bar{x},\bar{y}}$.
    Then, $r_{x,y} = r_{\bar{x},y}$ and $r_{x,\bar{y}} = r_{\bar{x},\bar{y}}$, which shows the claim.

    Thus, assume that $m_1 \leq n$.
    Now,
    \begin{align*}
        \prog_{x,y} &\step m_{x+1}; \dots; m_y; [\tfunc]~\{n - 1\} \\
        \prog_{x,\bar{y}} &\step m_{x+1}; \dots; m_{y-1}; [\tfunc]~\{n - 1\}.
    \end{align*}
    By \cref{lem:det,lem:par},
	\begin{align*} 
	\prog_{x, y}& \step^{y-x} m_y;[\tfunc]~\{n'\},\\
	 \prog_{x, \bar{y}}& \step^{y-x} [\tfunc]~\{n'\},\\ \prog_{\bar{x},y}& \step^{y-x-1} m_y;[\tfunc]~\{n''\},\\ \prog_{\bar{x},\bar{y}}& \step^{y-x-1} [\tfunc]~\{n''\},\\
	\end{align*}
	such that $n' \leq n'' \leq n' + 1$. If $n' = n''$, then as above $r_{x,y} = r_{\bar{x},y}$ and $r_{x,\bar{y}} = r_{\bar{x},\bar{y}}$, and we are done.

    Assume, therefore, that $n'' = n' + 1$.
    If $m_y > n' + 1$, then $m_y;[\tfunc]~\{n'\} \step [\tfunc]~\{n'\}$ and $m_y;[\tfunc]~\{n''\} \step [\tfunc]~\{n''\}$. Then, $r_{x, y} = r_{x, \bar{y}}$ and $r_{\bar{x},y} = r_{\bar{x},\bar{y}}$, which shows the claim.
    Else, if $m_y = n' + 1$, then again $m_y;[\tfunc]~\{n'\} \step [\tfunc]~\{n'\}$ but $m_y;[\tfunc]~\{n''\} \step [\tfunc]~\{n'\}$.
    Thus $r_{x,y}=r_{\bar{x},y}=r_{x,\bar{y}}=f(n')$ and $r_{\bar{x}, \bar{y}}=f(n'+1)$. The latter term is larger by monotonicity of $\tfunc$ and we are done.
    Else, it must hold that $m_y \leq n'$, thus $m_y;[\tfunc]~\{n'\} \step [\tfunc]~\{n'-1\}$ and $m_y;[\tfunc]~\{n''\} \step [\tfunc]~\{n'\}$.
    Now, $r_{x,y} = \tfunc(n'-1)$, $r_{\bar{x},y} = r_{x,\bar{y}} = \tfunc(n')$ and $r_{\bar{x},\bar{y}} = \tfunc(n' + 1)$. The claim now follows from convexity.
\end{proof}

\subsection{Putting It All Together}

\begin{proof}[Proof of Theorem~\ref{thm:submod}]
    As outlined in the beginning of the proof, it suffices to show that \cref{eq:interviewsuper} holds for a fixed order of agents, for fixed set of agents $\Grou$ and other agents $\grou$ and $\groualt$, and for fixed coin flip indices $\coin_{\agen}$ for all agents $\agen$ in the prefix.

    By \cref{lem:monoconv}, the expected number of left-over jobs is a monotone and convex function $\tof$ in the number of available jobs on entering the suffix.
    Due to independence between the coin flips, the expectation conditioned on the coin flips is the evaluation result of $c_{\agen_1}; c_{\agen_2}; \dots; c_{\agen_k}; [\tof]~\{\jobn_{\loca \prof}\}$, where $\agen_1, \dots, \agen_k$ are the prefix agents in their order.
    Potentially flipping the roles of $x$ and $y$, the inequality then follows from \cref{lem:xy}.
\end{proof}

\section{Proof of Theorem~\ref{thm:coordination}}
\label{app:coordination}

The proof of the theorem follows rather easily from the following lemma. 

\begin{lemma}[\citealt{rabanca16}]
\label{lem:bipartite}
Given a bipartite graph $G=(A\cup B,E)$, let $f:2^A\rightarrow \mathbb{N}$ be the function that maps $S\subseteq A$ to the size of the maximum matching in the induced graph $G[S\cup B]$. Then $f$ is submodular.
\end{lemma}

\begin{proof}[Proof of Theorem~\ref{thm:coordination}]
For each agent $\agen$ and job $\job$, flip a coin with bias $p_{\agen\job}$, and fix the coin flips hereinafter. 
As before, the employment function is a convex combination of functions corresponding to these fixed coin flips, so it is sufficient to show that each of these functions is submodular. 

For each locality $\loca$, consider a bipartite graph $G_\ell=(N\cup J_\loca,E_\loca)$, where $J_\loca$ is the set of jobs in $\loca$, and $(\agen,\job)\in E_\loca$ if and only if the (previously fixed) coin flip corresponding to this edge succeeded. 
Let $f_\loca:2^N\rightarrow \mathbb{N}$ be the function that maps $N'\subseteq N$ to the size of the maximum matching in the induced graph $G_\loca[N'\cup J_\loca]$. 
By Lemma~\ref{lem:bipartite}, $f_\loca$ is submodular. 

Now denote the employment function (for fixed coin flips) by $g$. 
We can write $g=\sum_{\loca\in L} g_\loca$, where each $g_\loca$ is defined by 
$$g_\loca(S):=f_\loca(\{\agen\in N:\ (\agen,\loca)\in S\}) $$
for all subsets of agent-locality pairs $S\subseteq 2^{N\times L}$. 
The submodularity of $g_\ell$ follows directly from the submodularity of $f_\loca$, and therefore $g$ itself is submodular. 
\end{proof}

\section{Experiments}
\label{app:simulations}
In the following sections, we provide more information about our simulations, namely the two experiments described in the body of the paper and two additional ones.
Besides figures and interpretations, each section links to the IPython notebook of the corresponding experiment.
These notebooks show the exact code and output of our simulation runs, and they can be downloaded for interactive experimentation.

\subsection{Number of Localities}
\label{app:num_localities}
The first experiment varies the number of localities. Its setting and the resulting diagram in \cref{fig:numlocalities} are described in detail in \cref{sec:simulations}.
The corresponding notebook can be found at \glink{https://github.com/pgoelz/migration/blob/master/num\_localities.ipynb}.

\subsection{Number of Agents}
\label{app:num_agents}

The second experiment varies the number $n$ of agents.
This modification of the setting was also described in \cref{sec:simulations}: There are 10 localities, each with an equal cap of $n/10$.
This experiment produced \cref{fig:numagents}, but we also generated the following plots of the absolute utilities achieved by both mechanisms:
\begin{center}
    \includegraphics[width=.45\textwidth]{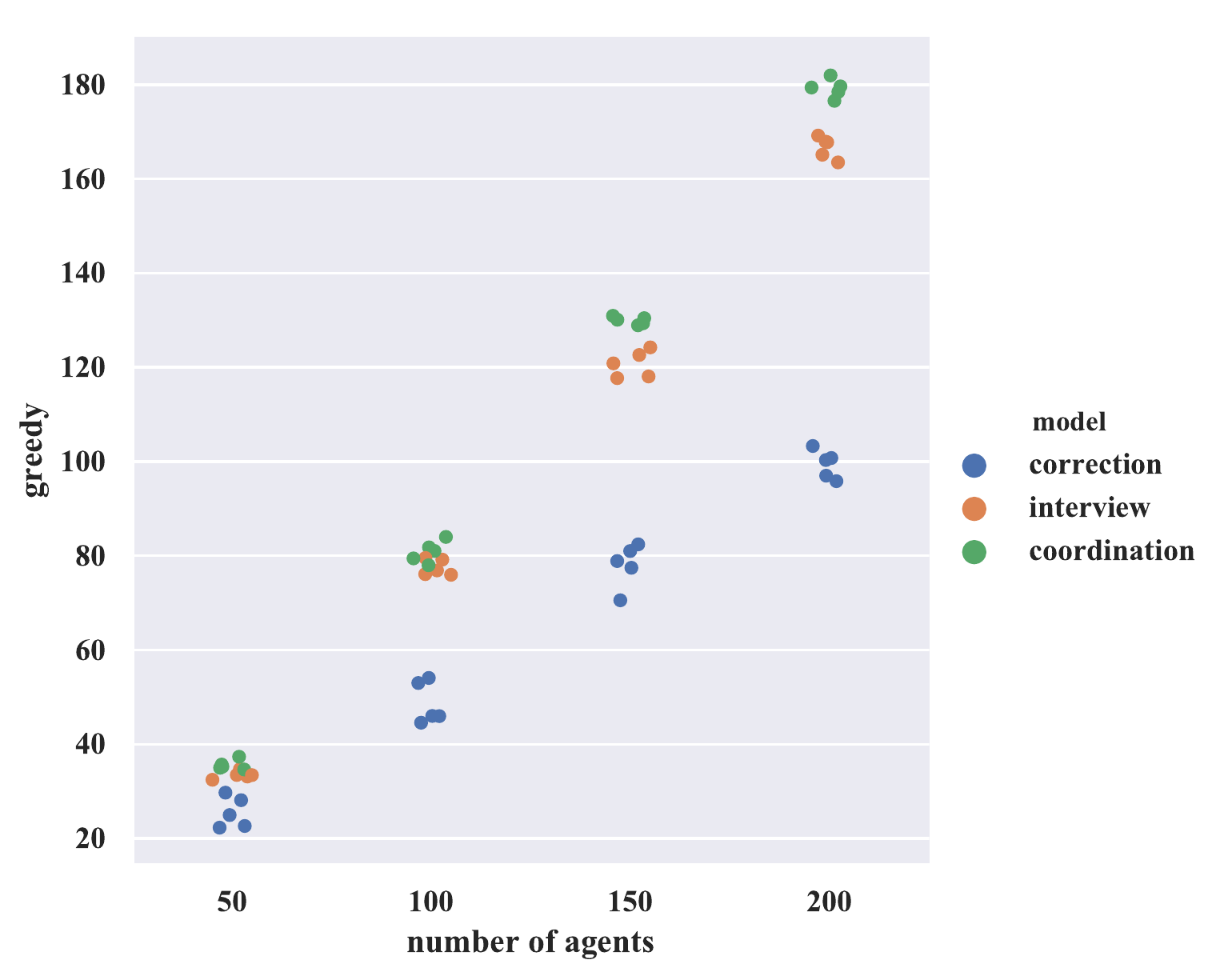} \quad
    \includegraphics[width=.45\textwidth]{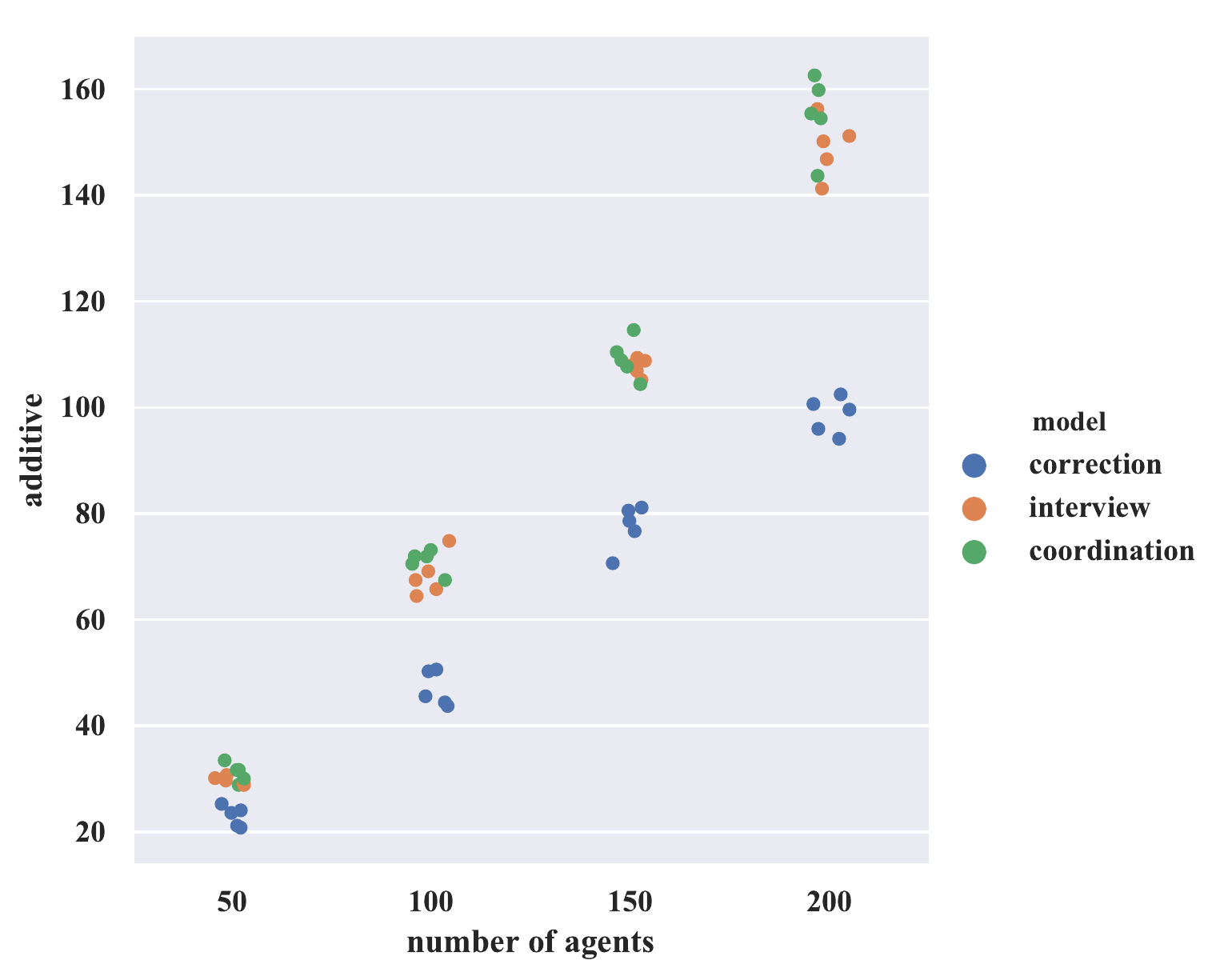}
\end{center}
Please mind the different scales of the y-axes.
As claimed before, the utility of the additive algorithm is very close to $n/2$ in the correction model, which is the expected value of the total number of people who can potentially be brought into work, since qualification probabilities are selected uniformly from $[0, 1]$.
This explains why the percentage improvement of the greedy algorithm over the additive one measured in \cref{fig:numagents} is relatively low in these settings.

The notebook for the experiment is available at \glink{https://github.com/pgoelz/migration/blob/master/num\_agents.ipynb}.

\subsection{Number of Professions}
\label{app:num_professions}
Let there be 10 localities, each with a capacity of 10.
The total number of agents is set back to 100.
We vary the number of professions, each of which has at least one agent.
All other agents are assigned to a uniformly chosen profession.
There is a job of the right profession for every agent, and each locality randomly receives 10 of these jobs.

\begin{center}
    \includegraphics[width=.45\textwidth]{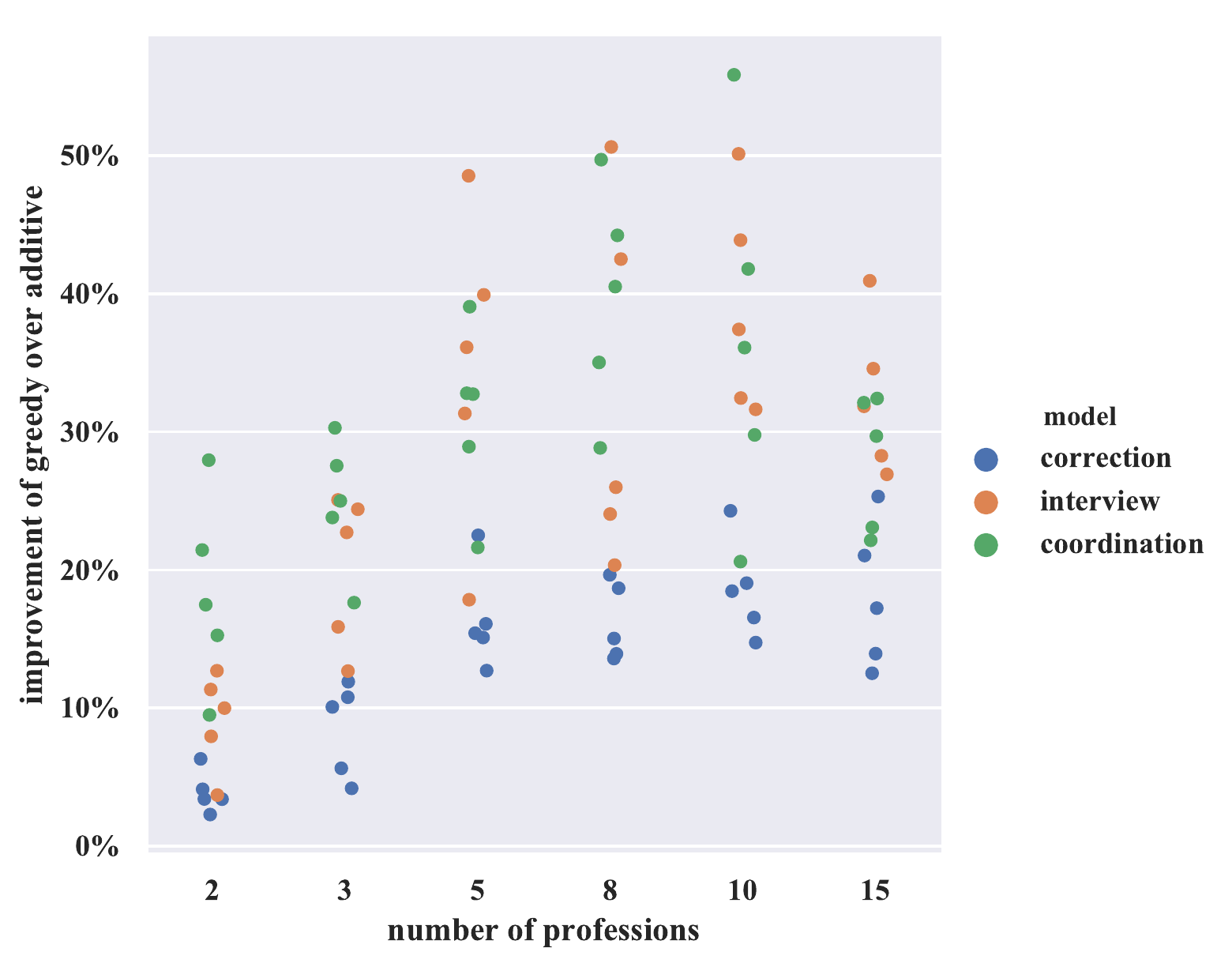}
\end{center}

When we increase the number of professions, the gap between greedy and additive increases significantly.
Our hypothesis is that, with many professions, a high fraction of matchings are bad just because the agents' professions do not line up with job availability, which disproportionately disadvantages the additive algorithm.
This can be observed in the following plot of absolute employment, where the performance of the additive algorithm stagnates on a low level, nearly independently of the model and at numbers of professions at which the greedy algorithm still achieves much higher employment.

\begin{center}
    \includegraphics[width=.45\textwidth]{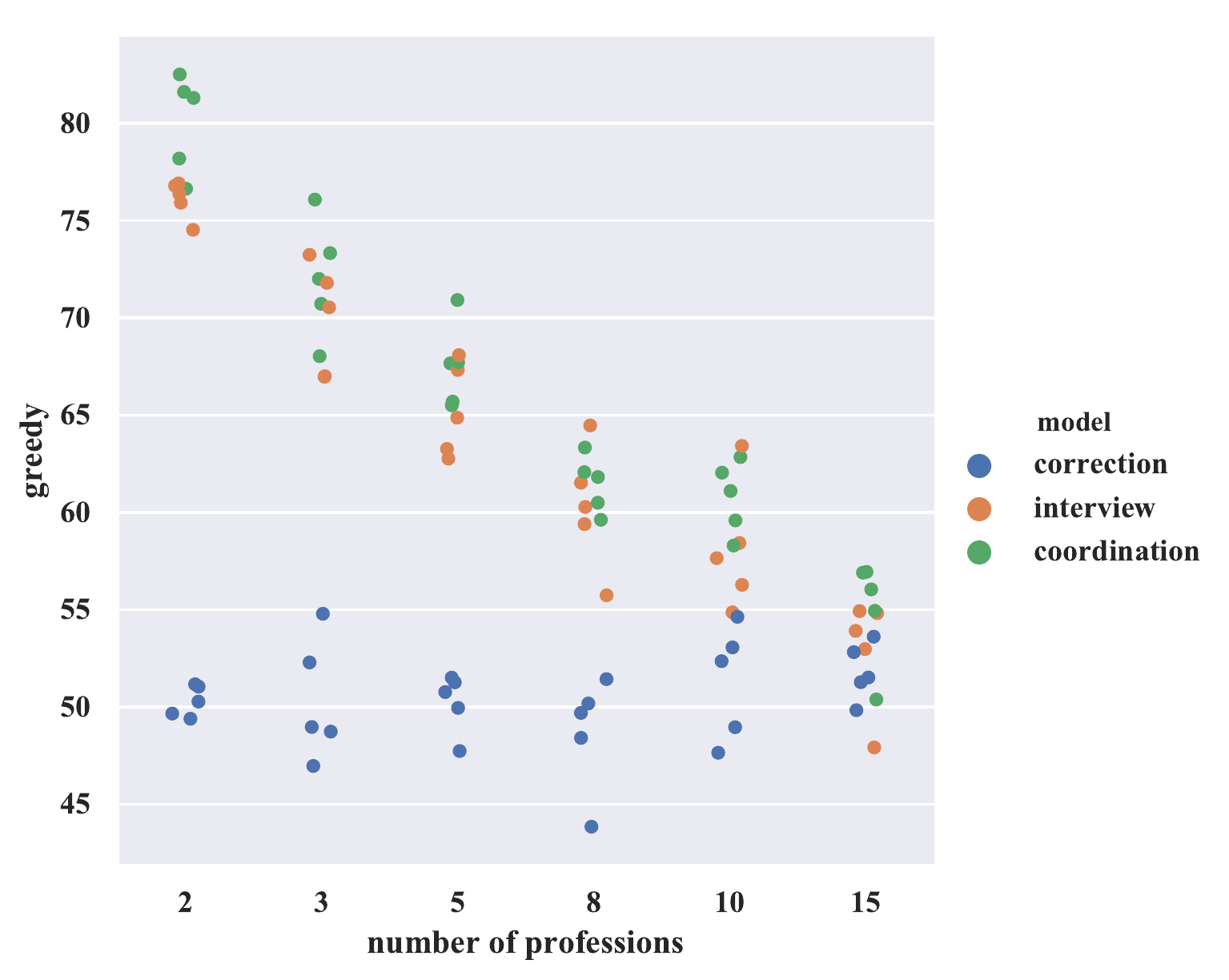} \quad
    \includegraphics[width=.45\textwidth]{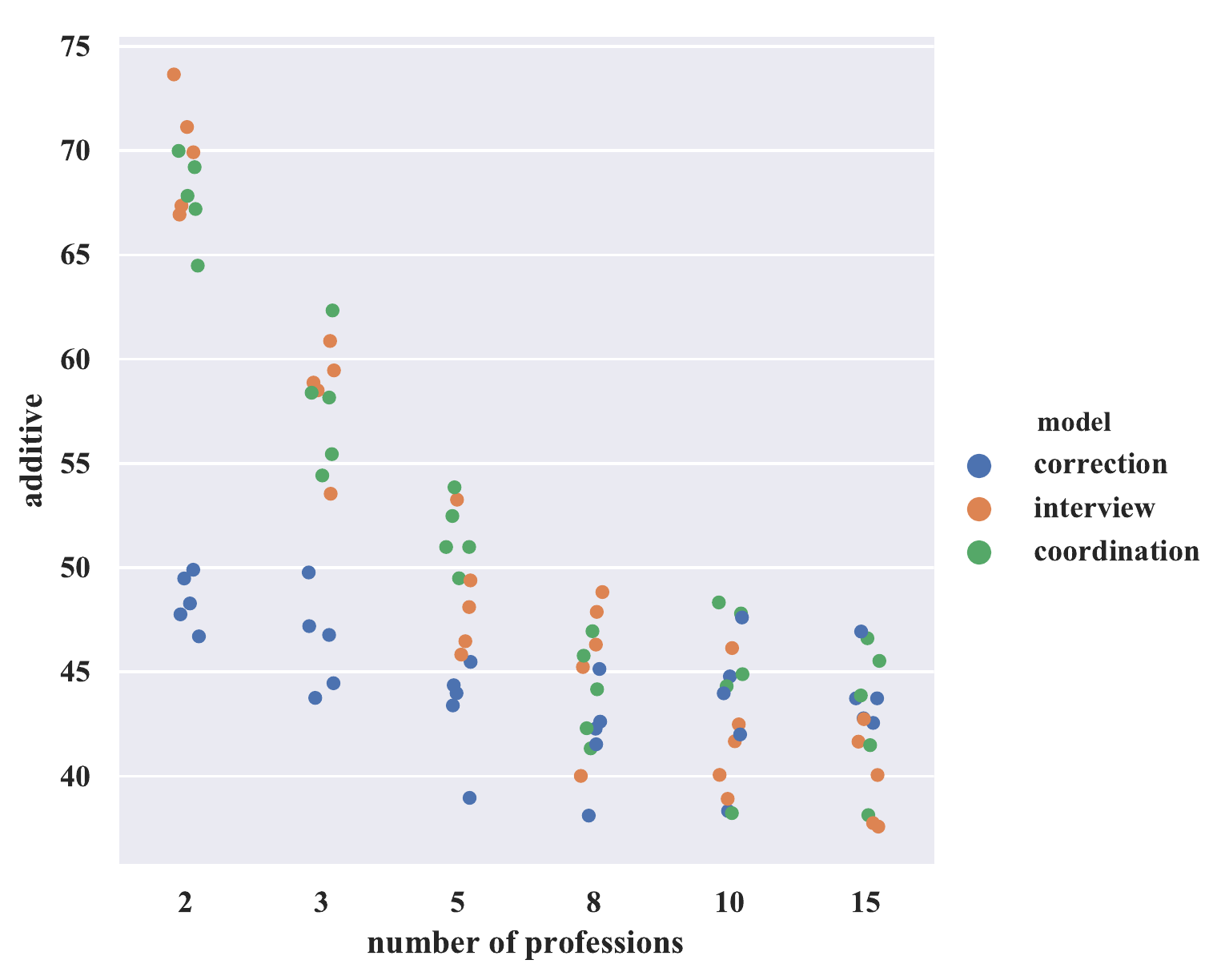}
\end{center}

The corresponding notebook can be found at \glink{https://github.com/pgoelz/migration/blob/master/num\_professions.ipynb}.

\subsection{Job Availability}
\label{app:job_availability}
In the following experiment, we fix 10 localities, each with a capacity of 10.
We have 100 agents, with 50 agents for each of the two professions.
Independently of the caps, we vary the number of jobs for both professions, which we distribute randomly between the localities.
\begin{center}
    \includegraphics[width=\textwidth]{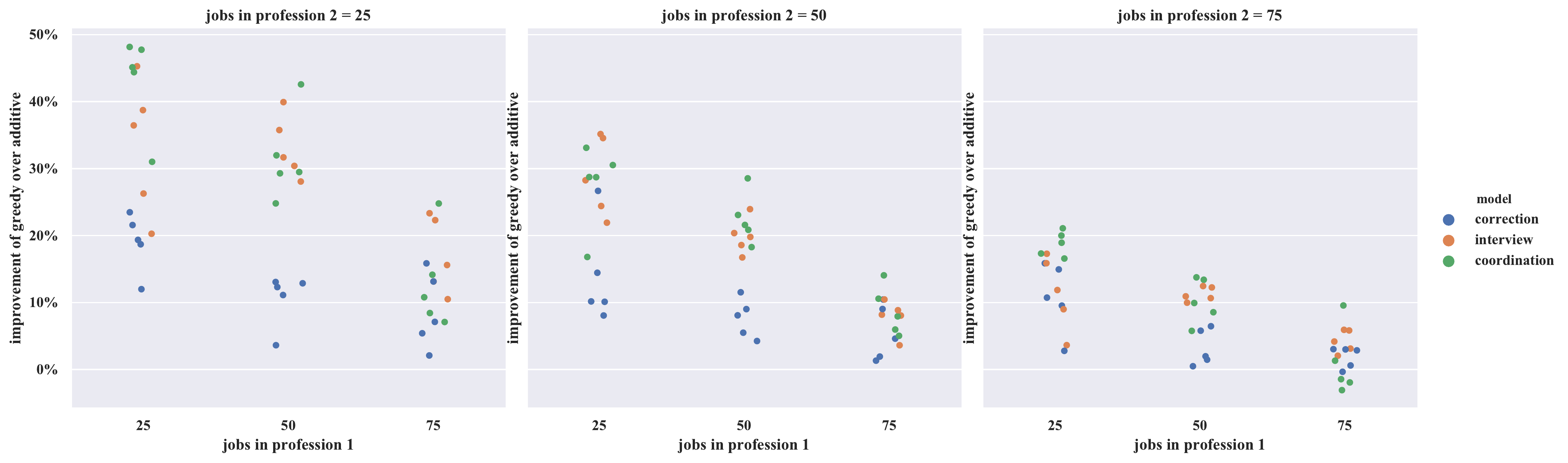}
\end{center}
Especially when jobs are scarce, i.e., when at least one profession has only 25 jobs for 50 agents, the greedy approach greatly improves upon the additive one.
When there is a surplus of both kinds of jobs, the gains again decrease since additive optimization already performs well in these situations, as can be seen in the following plot of absolute utilities:
\begin{center}
    \includegraphics[width=\textwidth]{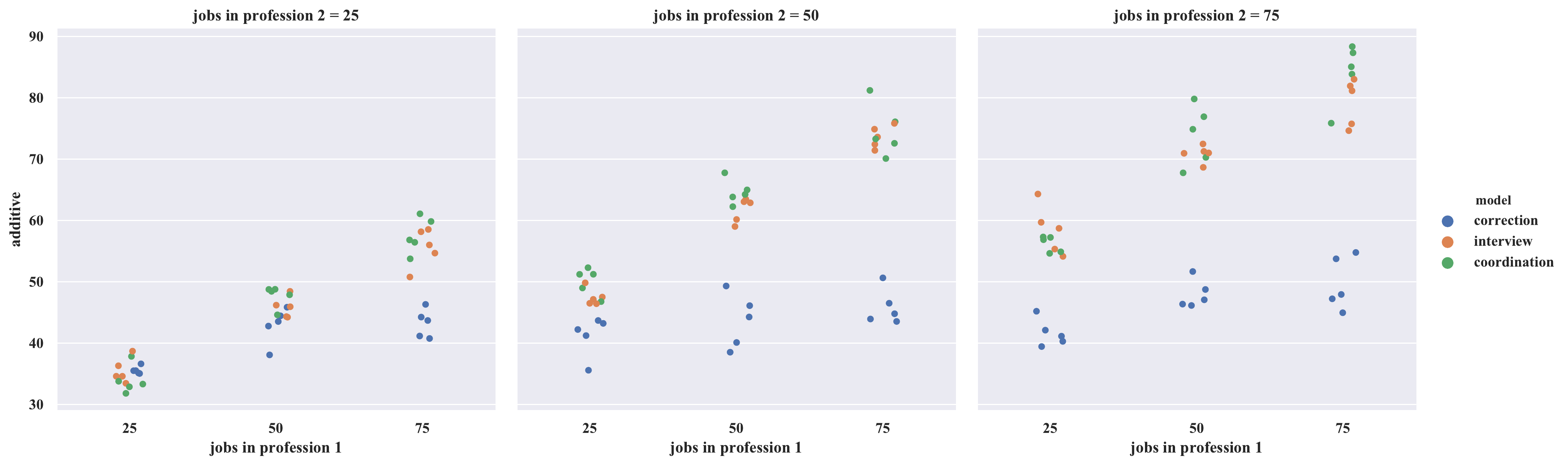}
\end{center}
In the correction model, the expected number of people qualifying for employment continues to be 50, which is why employment barely increases when going from 125 total jobs to 150.

Please refer to \glink{https://github.com/pgoelz/migration/blob/master/job\_availability.ipynb} for the simulation notebook.

\subsection{Specialization}
\label{app:specialization}
Our last setting again has 100 agents and 10 localities of capacity 10.
The total number of jobs is also 10 per locality now, but we do not distribute the jobs of the two professions randomly.

Instead, we define two kinds of localities, \emph{specialized} and \emph{unspecialized} ones.
In an unspecialized locality, both professions have 5 jobs each.
In a specialized locality, one profession has 8 jobs, the other just 2.
To keep the total number of jobs equal between both professions, we always have $s$ many localities specialized on profession~1 and $s$ many specialized on profession~2, where the \emph{specialization} $s$ is an integer between 0 and 5.
\begin{center}
    \includegraphics[width=.45\textwidth]{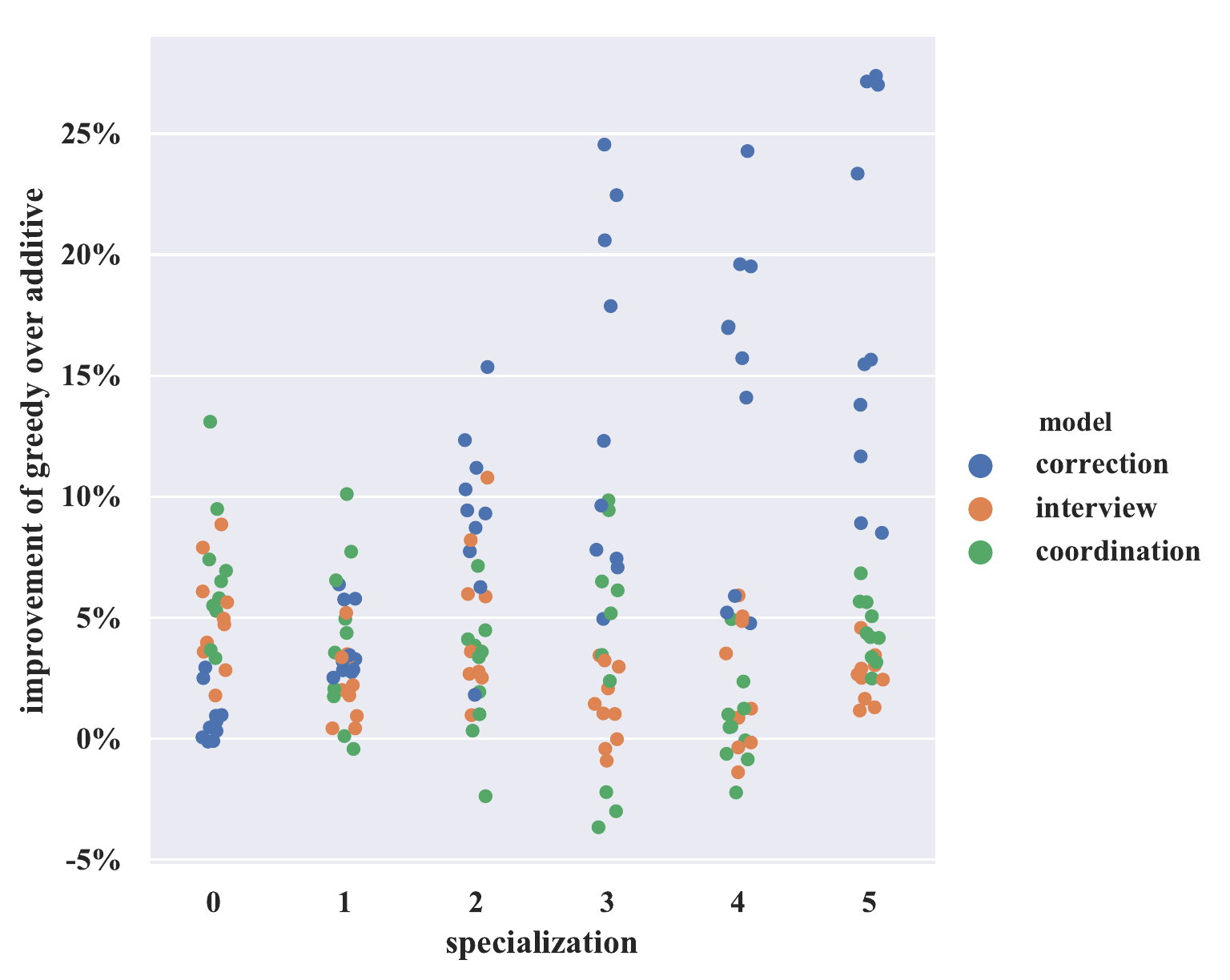}
\end{center}

In general, this setting seems to lead to less pronounced gains in terms of utility.
Interestingly, by far the biggest improvements can be made in the correction model, probably because the additive algorithm performs particularly bad in this setting.
For specialization levels of 3 and 4, we actually see a few cases where the additive algorithm performs slightly better than the greedy one.
Still, a majority of simulations in this setting --- whose only randomness comes from the agents' probabilities and the sampling noise --- see an improvement in utility, and the potential gains are stronger than the potential losses.

We direct the reader to \glink{https://github.com/pgoelz/migration/blob/master/specialization.ipynb} for the corresponding notebook.

\typeout{get arXiv to do 4 passes: Label(s) may have changed. Rerun}
 
\end{document}